\documentclass[onefignum,onetabnum]{siamart190516}

\usepackage[title]{appendix}

\usepackage{lineno,hyperref}
\usepackage{graphicx}
\usepackage{subcaption}
\usepackage{tikz}
\usetikzlibrary{arrows,calc}
\usetikzlibrary{patterns}
\tikzset{
	>=stealth',
	help lines/.style={dashed, thick},
	axis/.style={<->},
	important line/.style={thick},
	connection/.style={dotted},
}

\def\x{\mathbf{x}}

\def\0{\mathbf{0}}
\def\eps{\epsilon}

\def\gT{\mathcal{T}}
\def\lT{L}
\def\lC{L'}

\def\maxB{FBS }

\def\IP{(IP-FBS) }

\def\cC{\mathcal{C}}
\def\sj{^{(j)}}
\def\sw{^{(w)}}
\def\sjw{^{(j(w))}}
\def\dw{'^{(w)}}
\def\MFB{$2$-covariate BM problem}

\title{Algorithms and Complexity for Variants of Covariates Fine Balance }

\author{Dorit S. Hochbaum\thanks{Department of IEOR, Etcheverry Hall, Berkeley, CA, supported in part by NSF award No. CMMI-1760102.
		(\email{hochbaum@ieor.berkeley.edu}).}
	\and Asaf Levin\thanks{Faculty of Industrial Engineering and Management, The Technion, Haifa, Israel, supported in part by ISF - Israeli Science Foundation grant number 308/18.
		(\email{levinas@technion.ac.il}).}
	\and Xu Rao\thanks{Department of IEOR, Etcheverry Hall, Berkeley, CA, supported in part by NSF award No. CMMI-1760102.
		(\email{xrao@berkeley.edu}).}}
	
\headers{}{D. S. Hochbaum, A. Levin, and X. Rao}

\ifpdf
\hypersetup{
	pdftitle={Algorithms and Complexity for Variants of Covariates Fine Balance},
	pdfauthor={D. S. Hochbaum, A. Levin, and X. Rao}
}
\fi

\begin{document}
	
	\maketitle
	
	\begin{abstract}
We study here several variants of the covariates fine balance problem where we generalize some of these problems and introduce a number of others.  We present here a comprehensive complexity study of the covariates problems providing polynomial time algorithms, or a proof of NP-hardness.  The polynomial time algorithms described are mostly combinatorial and rely on network flow techniques.  In addition we present several fixed-parameter tractable results for problems where the number of covariates and the number of levels of each covariate are seen as a parameter.
	\end{abstract}
	
	

\section{Introduction}

The problem of balancing covariates arises in observational studies in various contexts such as statistics \cite{Rosenbaum02} \cite{Rubin06}, epidemiology \cite{Brookhart06}, sociology \cite{Morgan06}, economics \cite{Imbens04} and political science \cite{Ho07}. In an observational study there are two disjoint groups of samples, one  of treatment samples and the other of control samples.  Each of the samples in the two groups is characterized by several observed covariates, or features.

When estimating causal effects using observational data, it is desirable to replicate a randomized experiment as closely as possible by obtaining treatment and control groups with similar covariate distributions. This goal can often be achieved by choosing well-matched samples of the original treatment and control groups, thereby reducing bias in the estimated treatment effects due to the observed covariates.
The matching is to assign each treatment sample to one unique control sample, or, in other setups, to assign each treatment sample to a unique set of $\kappa$ control samples, for $\kappa$ a pre-specified integer, where every control sample is assigned to at most one treatment sample. A detailed review of matching-related methods used for covariates balancing problems is given by \cite{Stuart10}.

In this paper we address various problems of balancing covariates.  The covariates here are {\em nominal}, in that they take on discrete values or categories. The set of values of each nominal covariate partitions the treatment and control samples to a number of subsets referred to as \textit{levels} where the samples at every level share the same covariate value.  In an ideal situation the samples of the treatment and the control in each matched pair or matched set belong to the same levels over all covariates.  However, satisfying the requirement that matched samples in each pair or set belong to the same levels over all covariates typically results in a very small selection from the treatment and control group, which is not desirable.
To address this Rosenbaum et al. \cite{Rosenbaum07} introduced a weaker requirement to match all treatment samples to a subset of the control samples, called selection, so that the proportion (or the number, if $\kappa =1$)
of control and treatment samples in each level of each covariate are the same.  This requirement is known in the literature as {\em fine balance}.

To formalize the discussion we introduce essential notation.  Let the number of treatment samples be $n$ and the number of control samples be $n'$.  Let the set of all treatment samples be denoted by $\gT$, $|\gT |=n$.
Let $P$ be the number of covariates to be balanced. For $p=1,...,P$, covariate $p$ partitions both treatment and control groups into $k_p$ levels each.
Let the partition of the treatment group under covariate $p$ be  $L_{p,1}, L_{p,2}, ..., L_{p,k_p}$ of sizes $\ell_{p,1}, \ell_{p,2}, ..., \ell_{p,k_p}$. Similarly, let the partition of the control group under covariate $p$ be  $L'_{p,1}, L'_{p,2}, ..., L'_{p,k_p}$ of sizes $\ell'_{p,1}, \ell'_{p,2}, ..., \ell'_{p,k_p}$.
Let $\kappa$ be an integer specifying the ratio of the number of matched control samples to the number of matched treatment samples.

We define the $\kappa$-fine-balance constraints for a selection of treatment and a selection of control samples as follows:

\begin{definition}[$\kappa$-fine-balance]\label{def:fb}
For an integer $\kappa$, a selection $S \subseteq \gT$ of the treatment group and a selection $S'$ of the control group, we say that $(S,S')$-$\kappa$-fine-balance is satisfied if $\kappa \cdot |S\cap L_{p,i}|= |S'\cap L'_{p,i}|$ for $p=1,...,P$ and $i=1,...,k_p$.
\end{definition}
Obviously for $S$, $S'$ satisfying $(S,S')$-$\kappa$-fine-balance, the cardinality of $S'$ is $\kappa$ times as large as the cardinality of $S$, $|S'|=\kappa |S|$.  

We are now ready to define the three families of problems investigated here with complexity that varies according to the number of covariates and the value of $\kappa$.   The  {\em maximum $\kappa$-fine-balance selection} ($\kappa$-FBS) problem is to select a subset $S \subseteq \gT$ and a subset $S'$ of the control group so as to maximize the size of the selection $S$ (equivalent to maximizing the size of $S'$ since $|S'|=\kappa |S|$) where the $(S,S')$-$\kappa$-fine-balance constraints are satisfied.  This problem is introduced here for the first time.

A second problem studied here is the {\em $\kappa$-fine-balance matching} ($\kappa$-BM) problem, first introduced by Rosenbaum et al. \cite{Rosenbaum07}  for one covariate.  Here we are given a distance, or cost, measure between each treatment and each control sample.
The $\kappa$-BM problem is to minimize the total cost of the assignment of each treatment sample in $\gT$ to $\kappa$ control samples such that the selection of matched control samples $S'$ satisfies $(\gT,S')$-$\kappa$-fine-balance. 

Another problem family newly introduced here is an optimization where the feasible sets are optimal for another problem.   Formally, in the first stage the goal is to find the optimal selections to the $\kappa$-\maxB problem.
In the second stage, among all maximum sized selections, find the selection that minimizes the total distance of an assignment of each selected treatment sample to exactly $\kappa$ selected control samples. We refer to this problem as {\em maximum selection $\kappa$-fine-balance matching problem} ($\kappa$-MSBM).

For the case of $\kappa=1$, we ignore the prefix $\kappa$ so $(S,S')$-$\kappa$-fine-balance is called $(S,S')$-fine-balance,  $\kappa$-\maxB problem is called \maxB problem, $\kappa$-BM problem is called BM problem, and $\kappa$-MSBM problem is called MSBM problem.

A summary of the problems investigated here is given is Table \ref{table:problems}.

\begin{table}[h!]
	\centering
	\caption{Summary of problems studied here.}
	\begin{tabular}{ |c|c|c| }
		\hline
		Problem name & Objective & Constraints \\
		\hline
		max fine-balance selection (FBS) & $\max |S|$ & $(S,S')$-fine-balance \\
		max $\kappa$-fine-balance selection ($\kappa$-FBS)  & $\max |S|$ & $(S,S')$-$\kappa$-fine-balance \\
		fine-balance matching (BM) & $\min$ assignment cost & $(\gT,S')$-fine-balance \\
		$\kappa$-fine-balance matching ($\kappa$-BM) & $\min$ assignment cost  & $(\gT,S')$-$\kappa$-fine-balance \\
		max selection fine-balance matching (MSBM)& $\min$ assignment cost &  $(S,S')$ optimal for FBS\\
		max selection $\kappa$-fine-balance matching ($\kappa$-MSBM)& $\min$ assignment cost &  $(S,S')$ optimal for $\kappa$-FBS\\
		\hline
	\end{tabular}
	\label{table:problems}
\end{table}

\subsection{Related literature}
The concept of fine balance was first introduced by Rosenbaum et al. \cite{Rosenbaum07}, who studied the $\kappa$-BM problem for the $1$-covariate problem and proposed a network flow algorithm. No polynomial running time algorithm has been known for the $\kappa$-BM problem with two or more covariates.

It is not always feasible to find a selection $S'$ of the control samples that satisfies the $(\gT,S')$-$\kappa$-fine-balance constraints in the $\kappa$-BM problem. To that end several papers considered the goal of minimizing the violation of this requirement, which we refer to as {\em imbalance}, \cite{Yang12}, \cite{Zubizarreta12}, \cite{Pimentel15},  \cite{BVZ20}, \cite{HR20}. The studies in all these papers require the entire treatment group to be selected or matched. Bennett et al. \cite{BVZ20}, and Hochbaum and Rao \cite{HR20} considered finding the selection of control group that minimizes an imbalance objective, defined as $\sum_{p=1}^P \sum_{i=1}^{k_p} | |S'\cap L'_{p,i}|-\kappa\cdot\ell_{p,i} |$. This problem is called {\em minimum  $\kappa$-imbalance problem}. The problem is trivial to solve for the $1$-covariate problem (see \cref{sec:prelim} for details); the $2$-covariate problem was proved to be polynomial time solvable using linear programming in \cite{BVZ20} and using network flow algorithms in \cite{HR20}; for three or more covariates, the problem is NP-hard \cite{BVZ20,HR20}. Yang et al. \cite{Yang12} and Pimental et al. \cite{Pimentel15} considered a more complicated problem that minimizes the total assignment cost of the matched sets, each consisting of a single treatment sample and $\kappa$ control samples, subject to the requirement that the selection of matched control samples is optimal for the minimum $\kappa$-imbalance problem. Yang et al. \cite{Yang12} proposed two network flow algorithms for the case of the $1$-covariate problem; Pimental et al. \cite{Pimentel15} proposed a network flow algorithm for the case in which the covariates form a nested sequence.
Zubizarreta \cite{Zubizarreta12} considered a different variant which minimizes the total assignment cost of the matched sets with a penalty on the imbalance, and presented a mixed integer programming formulation for an arbitrary number of covariates.

 \subsection{Contributions}

We show here, for the first time, that the $2$-covariate BM problem is in fact NP-hard and therefore there is no polynomial time algorithm for the $\kappa$-BM problem with two or more covariates unless $P=NP$.  This NP-hardness result for $\kappa$-BM problem with two or more covariates holds for any value of $\kappa$.

The $\kappa$-MSBM problem is newly introduced here. It relaxes the requirement in the $\kappa$-BM problem of selecting {\em all} treatment samples and replaces it with a maximum size selection possible, while enforcing the $\kappa$-fine-balance constraints. This $\kappa$-MSBM problem, as shown here, is NP-hard with two or more covariates for any given value of $\kappa$. What's more, it is also proved here to be NP-hard for the $1$-covariate problem when $\kappa\geq 3$. We present a polynomial algorithm for the $1$-covariate MSBM problem, but we leave the complexity status of the $1$-covariate 2-MSBM problem open.  

The $\kappa$-\maxB problem is a simpler problem compared to the $\kappa$-MSBM problem. This problem is studied here for the first time. We prove that for three or more covariates, the \maxB and $\kappa$-\maxB problems are NP-hard for any value of $\kappa$.
For the case of the $2$-covariate problem we present here an efficient algorithm for the \maxB problem.  The algorithm is based on an integer programming formulation of the problem in which the constraint matrix, for two covariates, has the structure of network flow constraints.
For the resulting minimum cost network flow problem we apply an algorithm with running time $O(n\cdot (\min\{n+n', k_1k_2 \} + (k_1+k_2)\log (k_1+k_2))$. We also prove that for $\kappa \geq 3$, the $2$-covariate $\kappa$-\maxB problem is NP-hard. For the remaining case in which $\kappa=2$ and the number of covariates is two, the complexity status of the  $2$-\maxB problem is left open.

We observe here that, for {\em any} number of covariates, if the selections of treatment and control samples are fixed, then the optimal assignment among the selected samples, and therefore the optimal solution to the $\kappa$-MSBM problem, is attained by solving a minimum cost network flow problem. See \cref{sec:prelim} for details.

A summary of the complexity results for the three problem families is given in \cref{table:results}.

\begin{table}[h!]
	\centering
	\caption{Summary of complexity and algorithmic results derived here.}
	\subcaption*{(Here $n$ is the size of treatment group and $n'$ is the size of control group)}
	\begin{tabular}{ |c|c|c|c| }
		\hline
		Problem & One covariate & Two covariates & $\geq 3$ covariates \\
		\hline
		FBS & $O(n+n')$ & $O(n\cdot (\min\{n+n', k_1k_2 \} + (k_1+k_2)\log (k_1+k_2)))$  & NP-hard \\
		$\kappa$-FBS ($\kappa\geq 2$)   &  $O(n+n')$&  NP-hard for $\kappa\geq 3$, open for $\kappa=2$  & NP-hard\\
		$\kappa$-BM (any $\kappa$)& $O((n+n')^3)$ \cite{Rosenbaum07} & NP-hard & NP-hard \\
		MSBM & $O((n+n')nn')$ & NP-hard & NP-hard \\
		$\kappa$-MSBM ($\kappa\geq 2$)& NP-hard for $\kappa\geq 3$  & NP-hard & NP-hard \\
		& open for $\kappa=2$& & \\
		\hline
	\end{tabular}
	\label{table:results}
\end{table}

Beyond these complexity results, we also address here {\em fixed-parameter tractable} (FPT) results. We prove that the $\kappa$-FBS, $\kappa$-BM and MSBM problems are solvable in fixed-parameter tractable time for constant numbers of covariates levels, yet the $\kappa$-MSBM problem is NP-hard for constant $\kappa\geq 3$ even when the numbers of covariates levels are constant.  It remains an open problem whether the 2-MSBM problem is NP-hard for constant numbers of covariates levels.  

We also consider the $2$-covariate BM problem and the $2$-covariate $\kappa$-BM problem where one of the covariates has a constant number of levels (whereas the other one may have a linear number of levels), and we show that the complexity status of these special cases is tied to the complexity status of the exact matching problem, a problem that is known to have a randomized polynomial time algorithm \cite{MVV} but for which the existence of a deterministic polynomial time algorithm is a long-standing open problem.


\subsection{Paper overview}
In \cref{sec:prelim} we consider the case of the $1$-covariate $\kappa$-FBS, $\kappa$-BM, MSBM problems, and provide a compact representation of the sample selections. Then we present our complexity and algorithmic results of the other cases for the three families of problems separately, i.e., the $\kappa$-FBS problem in \cref{sec:maxB}, the $\kappa$-BM problem in \cref{sec:kBM}, and the $\kappa$-MSBM problem in \cref{sec:MSBM}. The fixed-parameter complexity results are provided in \cref{sec:FPT}.

\section{Preliminaries} \label{sec:prelim}

%

Consider first the case of a single covariate, $P=1$, that partitions the control and treatment groups into, say, $k$ levels each. Let the sizes of levels of the treatment group be $\ell_1, ..., \ell_k$, and the sizes of levels of the control group be $\ell'_1, ..., \ell'_k$.
It is easy to see that there exists a selection $S'$ of control samples that satisfies the $(\gT,S')$-$\kappa$-fine-balance if and only if $\ell'_i \geq \kappa \ell_i $ for $i=1,...,k$. If this condition is satisfied then any subset $S^*$ of the control group with $\kappa\cdot\ell_i$ samples in level $i$, $i=1,...,k$, satisfies the $(\gT,S^*)$-$\kappa$-fine-balance, and as such is a feasible selection for the $\kappa$-BM problem.
With these known numbers of control samples to be selected in each level, the optimal solution to the $1$-covariate $\kappa$-BM problem is found using a minimum cost network flow formulation, as shown next.
Note that a standard linear programming formulation of the minimum cost network flow (MCNF) is given in Appendix \ref{sec:flow}.

The MCNF problem the solution to which is an optimal solution to $\kappa$-BM is constructed on a bipartite graph with the treatment samples each represented by a node on one side, and the control samples each represented by a node on the other side.  The cost on each arc between a treatment sample and a control sample is the ``distance" value between the two, and the arc capacity is $1$.
Each treatment sample has a supply of $\kappa$.
To account for the requirement that in each level $i$ of control samples there will be $\kappa\cdot\ell_i$ samples matched we add to the bipartite graph a third layer of $k$ nodes, one for each level.  The $i$th node in the third layer  has demand of $\kappa\cdot\ell_i$ and there are arcs to this demand node from all control samples in level $i$  with capacity $1$ and cost of $0$.  In an optimal solution to this MCNF problem the control sample nodes through which there is a positive flow (of one unit) are the ones selected and matched to the respective treatment sample nodes from which they have a positive flow.

If $\ell'_i < \kappa \ell_i$ for some $i$ then there is no selection $S'$ of control samples that satisfies the $(\gT,S')$-$\kappa$-fine-balance.  Addressing this context, as mentioned earlier,  \cite{Yang12}, \cite{Zubizarreta12}, \cite{Pimentel15},  \cite{BVZ20}, \cite{HR20}, consider the problem of minimizing the $\kappa$-imbalance, which is the sum of violations for all levels,  $\sum_{i=1}^{k} | |S'\cap L'_{i}|-\kappa\cdot\ell_{i} |$.
The solution to this $1$-covariate minimum $\kappa$-imbalance problem is straightforward: in step 1, select $\min\{\kappa\cdot\ell_i,  \ell'_i \}$ control samples in level $i$ ; if the number of control samples selected is less than $n$ in step 1, then we select random additional control samples such that the selection is of size $n$.  Another way to address this context is to seek a solution for the $(S,S')$-$\kappa$-fine-balance where rather than forcing all samples of $\gT$ to be included, finding a solution in which the size of the selection $S$, and equivalently $|S'|$, is maximized, the $\kappa$-FBS problem. The solution to the $1$-covariate $\kappa$-FBS problem is also straightforward: select $\bar{\ell}_i=\min\{\ell_i, \lfloor \ell'_i/\kappa \rfloor \}$ treatment samples of level $i$ and $\kappa\cdot\bar{\ell}_i $ control samples of level $i$.

Again, for the $1$-covariate problem of finding an optimal matching, or assignment, among all optimal selections for either the minimum $\kappa$-imbalance or the FBS problem, we solve a MCNF problem for the known number of samples to select from each level, similar to the one defined above with the following modifications.  For the minimum $\kappa$-imbalance, we first need to change the demand of the demand nodes in the above MCNF problem from $\kappa\cdot\ell_i$ to $\min\{\kappa\cdot\ell_i,  \ell'_i \}$ for each level $i$. We also add a dummy demand node in the third layer with demand $\kappa \cdot n-\sum_{i=1}^k \min\{\kappa\cdot\ell_i,  \ell'_i \}$, which connects with all control nodes each with capacity 1 and cost of 0.
For the optimal selections of FBS, in addition to changing the demand from $\ell_i$ to $\bar{\ell}_i$ for each level $i$, we also remove the supply on each treatment sample, add for every level $i$ a supply node with supply $\bar{\ell}_i$, and add arcs from this supply node to all treatment samples in level $i$ with capacity $1$ and cost of $0$. The best assignment found with a selection that is optimal for the FBS problem is an optimal solution for the MSBM problem. However, this method does not apply to the $\kappa$-MSBM problem with $\kappa \geq 2$.  We further show that even the $1$-covariate $\kappa$-MSBM problem is NP-hard for $\kappa\geq 3$ (see \cref{sec:MSBM}).

Hence, all problems discussed here except for the $\kappa$-MSBM problem are polynomial time solvable for the $1$-covariate case.  In \cref{sec:MSBM} we show that the $1$-covariate $\kappa$-MSBM  does not admit a polynomial time algorithm for $\kappa\geq 3$ unless P=NP.

Consider next the case of multiple covariates.
For the $\kappa$-\maxB problem, we observe that the selections from the treatment and control groups can be represented compactly in terms of {\em level-intersections}.  For $P$ covariates, the intersection of the level sets $L_{1,i_1} \cap  L_{2,i_2}\cap \ldots \cap L_{P,i_P}$, $i_p=1, \ldots ,k_p$, $p=1, \ldots ,P$, form a partition of the treatment group. Similarly, the intersection of the level sets $L'_{1,i_1} \cap  L'_{2,i_2}\cap \ldots \cap L'_{P,i_P}$, $i_p=1, \ldots ,k_p$, $p=1, \ldots ,P$, form a partition of the control group.
Therefore, instead of specifying which sample belongs to the selection, it is sufficient to determine the {\em number} of selected samples in each level intersection for the two groups, since the identity of the specific selected samples has no effect on the fine balance requirement. With this discussion we have a theorem on the representation of the solution to the $\kappa$-\maxB problems in terms of the level-intersection sizes.

\begin{theorem}\label{thm:unique}
	The level-intersection sizes $s_{i_1,i_2,\ldots ,i_P}$ and $s'_{i_1,i_2,\ldots ,i_P}$ are an optimal solution to the $\kappa$-\maxB problem if there exists an optimal selection $S$ of treatment samples and $S'$ of control samples such that $s_{i_1,i_2,\ldots ,i_P}=|S \cap L_{1,i_1} \cap  L_{2,i_2}\cap \ldots \cap L_{P,i_P}|$ and $s'_{i_1,i_2,\ldots ,i_P}=|S' \cap L'_{1,i_1} \cap  L'_{2,i_2}\cap \ldots \cap L'_{P,i_P}|$, for $p=1, \ldots ,P$, $i_p=1, \ldots ,k_p$.
\end{theorem}
We will say that the optimal selection for the covariates problems here is {\em unique} if for any optimal selection $S$ and $S'$, the numbers $s_{i_1,i_2,\ldots ,i_P}=|S \cap L_{1,i_1} \cap  L_{2,i_2}\cap \ldots \cap L_{P,i_P}|$ and $s'_{i_1,i_2,\ldots ,i_P}=|S' \cap L'_{1,i_1} \cap  L'_{2,i_2}\cap \ldots \cap L'_{P,i_P}|$ are unique.
In order to derive an optimal selection given the optimal level-intersection sizes, one selects any $s_{i_1,i_2,\ldots ,i_P}$ treatment samples from the intersection $L_{1,i_1} \cap  L_{2,i_2}\cap \ldots \cap L_{P,i_P}$ and any $s'_{i_1,i_2,\ldots ,i_P}$ control samples from the intersection $L'_{1,i_1} \cap  L'_{2,i_2}\cap \ldots \cap L'_{P,i_P}$ for $i_p=1, \ldots ,k_p$, $p=1, \ldots ,P$. 

We observe here that, for any number of covariates, if the optimal selection of treatment and control samples in terms of level-intersections is known and unique, then the optimal assignment among the selected samples, and therefore the optimal solution to the $\kappa$-MSBM problem, can also be attained by solving an MCNF problem as follows.
For each non-zero level intersection of treatment samples there is a source node with supply of $ s_{i_1,i_2,\ldots ,i_P}$.  This source node is connected to all treatment samples in the intersection $L_{1,i_1} \cap  L_{2,i_2}\cap \ldots \cap L_{P,i_P}$ with arcs of capacity $1$ and cost of $0$.  For each non-zero level intersection of control samples there is a demand node with supply of $s'_{i_1,i_2,\ldots ,i_P}$. This demand node is connected from all control samples in the intersection $L'_{1,i_1} \cap  L'_{2,i_2}\cap \ldots \cap L'_{P,i_P}$ with arcs of capacity $1$ and cost of $0$. The treatment and control sample nodes through which there is a positive flow (of sone unit) are the ones selected, and a positive flow between a  treatment node and a control node indicates the two samples are matched.  This is a minimum cost network flow problem with a total demand (or supply) bounded by $\min\{n,n'\}$, and $O(nn')$ arcs and $O(n+n')$ nodes.  Therefore the successive shortest paths algorithm, discussed below in \cref{sec:maxB}, solves this problem in $O((n+n') nn')$ steps.

\section{The maximum  $\kappa$-fine-balance selection ($\kappa$-\maxB) problem}\label{sec:maxB}

In this section we show the complexity and algorithmic results for the $\kappa$-\maxB problems. We present the results separately first for the $1$-covariate problem, next for three or more covariates, then the $2$-covariate \maxB problem, and finally the $2$-covariate $\kappa$-\maxB problem for $\kappa\geq 3$.

%


The solution to the $1$-covariate  $\kappa$-\maxB problem is straightforward, as discussed in \cref{sec:prelim}: select $\bar{\ell}_i=\min\{\ell_{1,i}, \lfloor \ell'_{1,i}/\kappa \rfloor \}$ number of level $i$ treatment samples and $\kappa\cdot\bar{\ell}_i $ number of level $i$ control samples. The union of the selections at each level is an optimal solution for the $1$-covariate $\kappa$-\maxB problem.The value of the objective function corresponding to this solution is $\sum_{i=1}^{k_1} \bar{\ell}_i=\sum_{i=1}^{k_1} \min\{\ell_{1,i},  \lfloor \ell'_{1,i}/\kappa \rfloor\}$.

\subsection{NP-hardness for the $\kappa$-\maxB problem for any constant $\kappa$ with $P\geq3$}\label{sec:NPC}

We show here that even for $\kappa$ being constant, the $\kappa$-\maxB problem with three or more covariates is NP-hard by reducing the  {\em 3-Dimensional Matching} problem, one of Karp's 21 NP-hard problems \cite{KARP72}\cite{GJ79}.
\\
{\em 3-Dimensional Matching}: Given a finite set $X$ and a set of triplets $U\subset X\times X\times X$. Is there a subset $M\subseteq U$ such that $|M|=|X|$ and that no two elements of $M$ agree in any coordinate?

\begin{theorem}\label{thm:npP3} The $\kappa$-\maxB problem is NP-hard when $P=3$ even for constant $\kappa$.
\end{theorem}

\begin{proof}
	Given an instance of 3-dimensional matching problem with a finite set $X$ and a set of triplets $U\subset X\times X\times X$, we construct an instance of $\kappa$-\maxB problem with $P=3$ for any constant $\kappa$. Without loss of generality, we assume $X=\{1,...,|X|\}$.
	
	First we define the levels of the three covariates. For $p=1,2,3$, the set of levels of covariate $p$ is $\{1,...,|X|\} \cup \{0,0'\}$. So all the three covariates have $|X|+2$ levels.
	
	Next, we construct the samples. For each sample, we represent it by the ordered triplet $(a,b,c)$ where $a$ is the level of the first covariate, $b$ is the level of the second covariate, and $c$ is the level of the third covariate. The treatment group contains a sample $(i,i,i)$ for $i=1,...,|X|$ as well as $|X|$ copies of $(0,0,0)$ and $|X|$ copies of $(0',0',0')$.
	For each triplet $u\in U$, whose elements are denoted by $[u_1, {u_2}, {u_3}]$, we create one control sample $({u_1}, {u_2}, {u_3})$. In addition, for each element $i\in X$ we have $(\kappa-1)$ copies for each of the three control samples $(i,0',0), (0',0,i), (0,i,0')$. We also create $|X|$ copies of $(0,0,0)$ and $|X|$ copies of $(0',0',0')$ for the control group. That is, the control group is the union of the following three sets (we represent the different copies by a superscript as shown below.):
	\begin{eqnarray*}
	C_1 &=& \{({u_1}, {u_2}, {u_3}): \forall u=[{u_1}, {u_2}, {u_3}]\in U \}, \\
	C_2 &= &\{ (i,0',0)\sw, (0',0,i)\sw, (0,i,0')\sw: \forall i=1,...,|X|, \forall w=1,...,\kappa-1\}, \\
	C_3 &= &\{ (0,0,0)\sw, (0',0',0')\sw:\forall w=1,...,|X|  \} .
	\end{eqnarray*}
	The treatment group constructed is of size $3|X|$ and the control group constructed is of size $|U|+(3\kappa-1)|X|$. The two sizes are both polynomially bounded in the size of the 3-dimensional matching instance so the reduction can be computed in polynomial time.
	
	Finally, we claim that the optimal value of the constructed $3$-covariate $\kappa$-\maxB problem is $3|X|$ if and only if there exist a subset $M\subseteq U$ such that $|M|=|X|$ and that no two elements of $M$ agree in any coordinate for the 3-dimensional matching instance.
	
	Let $M\subseteq U$ be the solution for the 3-dimensional matching instance, we derive a solution for the constructed problem as follows. We select all the treatment samples. For each triplet $u\in M$, we choose the control sample in $C_1$ whose covariates levels are corresponding to the elements in $u$.
	Additionally, we also choose all control samples in $C_2$ and $C_3$.
	To check the feasibility of this solution, first consider the appearances of level $i$ for each $i=1,...,|X|$ under each covariate $p=1,2,3$: level $i$ appears once in the treatment group for each covariate $p$; it appears once among the selected samples in $C_1$ as $i$ appears once in each coordinate in $M$; it appears for $\kappa-1$ times in $C_2$; it does not appear in $C_3$. That is, there is exactly one selected treatment sample of level $i$ under covariate $p$ and $\kappa$ selected control samples of level $i$ under covariate $p$, for each $i$ and $p$.
	Next, consider appearances of level $0$ and $0'$ under each covariate $p=1,2,3$: they each appear $|X|$ times in the treatment group; they do not appear in $C_1$; they each appear $(\kappa-1) |X|$ times in $C_2$ and $|X|$ times in $C_3$. That is, $0$ and $0'$ each appear $\kappa |X|$ times in the selected control samples under each covariate.
	Therefore, this selection is feasible and the objective value, the number of selected treatment samples, is $3|X|$.
	
	On the other hand, if the constructed $3$-covariate $\kappa$-\maxB problem has an optimal solution $S,S'$ of objective value $3|X|$, we say that $u=[u_1, {u_2}, {u_3}]\in U$ is selected for $M$ if the control sample $(u_1, {u_2}, {u_3})$ is selected in $S'$ for the constructed problem.  We will show that $M$ is a feasible solution of the 3-dimensional matching instance. Since the size of the treatment group is $3|X|$, all treatment samples must be selected in the optimal solution, and that $3\kappa |X|$ number of control samples must be selected. For each covariate $1,2,3$, levels $0$ and $0'$ each appears $|X|$ times in the treatment group, so the number of appearance of each of these two levels must be $\kappa |X|$ in the selection $S'$ of the control samples. So all samples in $C_2$ and $C_3$ must be selected, otherwise there is no enough level $0$ or level $0'$ samples in $S'$. Therefore, $M=3\kappa |X| -|C_2| - |C_3|=|X|$. Furthermore, for each covariate $p$ and for $i=1,...,|X|$, level $i$ appears exactly once in the treatment group so there are $\kappa$ number of selected control samples in level $i$. For each $i$ under each covariate $p$, since there are $(\kappa-1)$ number of samples in level $i$ in $C_2\cup C_3$, only one sample in $C_1$ in that same level is selected. So there is no overlap in each coordinate for any two triplets in $M$.
	
	With the above arguments, any 3-dimensional matching problem can be reduced to a $3$-covariate $\kappa$-\maxB problem for any constant integer $\kappa$, and hence, the  $\kappa$-\maxB problem is NP-hard for any such $\kappa$ when $P=3$.
\end{proof}

\begin{corollary}\label{cor:np3cov}
	The $\kappa$-\maxB problem is NP-hard for any integer $P\geq3$, even for constant $\kappa$.
\end{corollary}
\begin{proof}
	
	For any constant integer $\kappa$, any $3$-covariate $\kappa$-\maxB problem, and any $P>3$, we can construct an equivalent $P$-covariate $\kappa$-\maxB problem as follows: for each sample of the given $3$-covariate  $\kappa$-\maxB problem, we create a sample for the constructed  $\kappa$-\maxB problem such that they have the same level value for covariate $p=1,2,3$. For $p=4,...,P$, set covariate $p$ to have only one level so all samples in the constructed  $\kappa$-\maxB problem have the same value.
	
	Therefore, the NP-hardness of $3$-covariate  $\kappa$-\maxB problem implies that the $P$-covariate $\kappa$-\maxB problem is NP-hard for every value of $P$ when $P\geq3$.
\end{proof}

Since the $\kappa$-\maxB problem is NP-hard for $P\geq 3$, there is no polynomial time algorithm unless $P=NP$.

In the following subsections, we will discuss the remaining case of the $2$-covariate problems.


\subsection{The network flow algorithm for \maxB with $P=2$}\label{sec:ip}

In this subsection, we present an integer programming formulation with network flow constraints for the $2$-covariate \maxB problem. We then show how to solve the problem efficiently with a network flow algorithm.

It was noted, in \cref{thm:unique}, that there is no differentiation between the individual samples selected in each level intersection, only the number of those selected counts. We thus define the decision variables as follows:
\\ \indent $x_{i_1,i_2}$: the number of  treatment samples selected from the $(i_1, i_2)$ level intersection $\lT_{1,i_1} \cap \lT_{2,i_2}$, for $ i_1=1,...,k_1$ and $ i_2=1,...,k_2$;
\\ \indent  $x'_{i_1,i_2}$: the number of  control samples selected from the $(i_1, i_2)$ level intersection $\lC_{1,i_1} \cap \lC_{2,i_2}$, for $ i_1=1,...,k_1$ and $ i_2=1,...,k_2$. \\
Let $u_{i_1,i_2}=|\lT_{1,i_1} \cap \lT_{2,i_2}|$ and $u'_{i_1,i_2}=|\lC_{1,i_1} \cap \lC_{2,i_2}|$ for $ i_1=1,...,k_1$, $ i_2=1,...,k_2$. Clearly, $x_{i_1,i_2}$ must be an integer between $0$ and $u_{i_1,i_2}$, and $x'_{i_1,i_2}$ must be an integer between $0$ and $u'_{i_1,i_2}$.
With these decision variables the following is an integer programming formulation for the $2$-covariate \maxB problem:
\begin{subequations}\label{MS}
	\begin{align}
	\mbox{\IP~~~~}& \max &  \sum_{i_1=1}^{k_1} \sum_{i_2=1}^{k_2}  x_{i_1,i_2} &\label{obj} \\
	&\text{s.t.} &  \sum_{i_2=1}^{k_2}  x_{i_1,i_2} - \sum_{i_2=1}^{k_2}  x'_{i_1,i_2}  =0   && i_1=1,...,k_1  \label{cons:b1} \\
	&&  \sum_{i_1=1}^{k_1}  x_{i_1,i_2} - \sum_{i_1=1}^{k_1}  x'_{i_1,i_2}  =0   && i_2=1,..., k_2 \label{cons:b2}\\
	&& 0\leq x_{i_1,i_2} \leq u_{i_1,i_2}&&  i_1=1,...,k_1, \quad i_2=1,..., k_2 \label{cons:t}\\
	&& 0\leq x'_{i_1,i_2} \leq u'_{i_1,i_2}&&  i_1=1,...,k_1, \quad i_2=1,..., k_2 \label{cons:c}\\
	&&  x_{i_1,i_2}, x'_{i_1,i_2} \text{ integers}  && i_1=1,...,k_1, \quad i_2=1,..., k_2 \label{cons:int}
	\end{align}.
\end{subequations}

The objective \cref{obj} is the total number of selected treatment samples. Constraints \cref{cons:b1} are the fine balance requirement under covariate 1, as $\sum_{i_2=1}^{k_2}  x_{i_1,i_2}$ equals the number of selected treatment samples in level $i_1$ under covariate 1 and $ \sum_{i_2=1}^{k_2}  x'_{i_1,i_2}$ equals the number of selected control samples in the same level. Similarly,  constraints \cref{cons:b2} are the fine balance requirement under covariate 2.

Formulation \IP is in fact also a network flow formulation. In a minimum cost network flow formulation, each column of the constraint matrix corresponding to a variable that is a flow along an arc, has exactly one 1 and one -1.
The corresponding MCNF network is shown in \cref{fig:mcnf}, where all capacity lower bounds are $0$, and each arc has a cost per unit flow and upper bound associated with it.
The flow on the arc from node $(1,i_1)$ to node $(2,i_2)$ represents variable $x_{i_1,i_2}$, which is bounded between $0$ and $u_{i_1,i_2}$ as stated in constraints \cref{cons:t}; arc from node $(2,i_2)$ to node $(1,i_1)$ represents variable $x'_{i_1,i_2}$, which is bounded between $0$ and $u'_{i_1,i_2}$ as stated in constraints \cref{cons:c}. To get a ``minimize'' type objective, we take the negative value of $|S|= \sum_{i_1=1}^{k_1} \sum_{i_2=1}^{k_2}  x_{i_1,i_2}$ as the objective, so the per unit arc cost should be $-1$ for arcs from any node in $\{(1,1), (1,2), ..., (1,k_1) \}$ to any node in $ \{(2,1), (2,2), ...,(2,k_2)\}$. All other arcs have cost $0$. It is easy to verify that constraints \cref{cons:b1} are corresponding to the flow balance at nodes $(1,i_1)$ for all $i_1$, constraints \cref{cons:b2} are corresponding to the flow balance at nodes $(2,i_2)$ for all $i_2$.

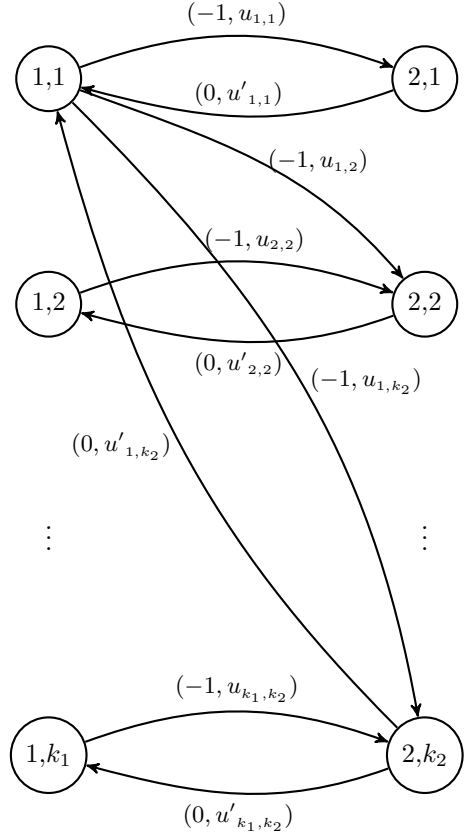
\begin{figure}[htb]
	\centering
	\begin{tikzpicture}	
	\begin{scope}[every node/.style={circle,thick,draw}]
	\node (11) at (0,0) {1,1};
	\node (12) at (0,-3) {1,2};
	\node (1k1) at (0,-9) {1,$k_1$};
	\node (21) at (5,0) {2,1};
	\node (22) at (5,-3) {2,2};
	\node (2k2) at (5,-9) {2,$k_2$};
	\end{scope}

	\begin{scope}[every node/.style={thick}]
	\node (1.) at (0,-6) {$\vdots$};
	\node (2.) at (5,-6) {$\vdots$};
	
	\end{scope}
	
	\begin{scope}[
	every edge/.style={draw=black, thick,font=\small,align=center, auto}]
	\path [->] (11) edge [out=20,in=160] node {$(-1, u${\tiny $_{1,1}$}$)$} (21);
	\path [->] (11) edge [out=-20,in=130] node[right] {$(-1, u${\tiny $_{1,2}$}$)$} (22);
	\path [->] (11) edge [out=-45,in=100] node[right] {$(-1, u${\tiny $_{1,k_2}$}$)$} (2k2);
	\path [->] (12) edge  [out=20,in=160] node[pos=.559] {$(-1, u${\tiny $_{2,2}$}$)$}  (22);
	\path [->] (1k1) edge [out=20,in=160] node {$(-1, u${\tiny $_{k_1,k_2}$}$)$}  (2k2);
	
	\path [->] (21) edge [out=200,in=-15] node[above] {$(0, u'${\tiny $_{1,1}$}$)$} (11);
	\path [->] (22) edge [out=200,in=-15] node {$(0, u'${\tiny $_{2,2}$}$)$}  (12);
	\path [->] (2k2) edge [out=200,in=-15] node {$(0, u'${\tiny $_{k_1,k_2}$}$)$}  (1k1);
	\path [->] (2k2) edge [out=135,in=-75] node[left] {$(0, u'${\tiny $_{1,k_2}$}$)$} (11);

	\end{scope}
	
	\end{tikzpicture}
	\caption{Min-cost network flow graph corresponding to formulation \IP.}
	\subcaption*{arc legend: (cost, upperbound)}
	\label{fig:mcnf}
\end{figure}

\begin{theorem}\label{thm:timeMCNF}
	The $2$-covariate \maxB problem is solved as a minimum cost network flow problem in $O(n\cdot (\min\{n+n', k_1k_2 \} + (k_1+k_2)\log (k_1+k_2))$ time.
\end{theorem}

\begin{proof}
	
	To solve the minimum cost network flow problem of the $2$-covariate \maxB problem, we choose the algorithm of {\em successive shortest paths} that is particularly efficient for a MCNF with ``small" total arc capacity (see \cite{AMO93} Section 9.7). The successive shortest paths algorithm starts with a network graph with no negative cycles, so we first modify the network shown in  \cref{fig:mcnf} using a well-known arc reversal transformation in \cite{AMO93} Section 2.4. The resulting network graph is shown in  \cref{fig:mcnf2}.

	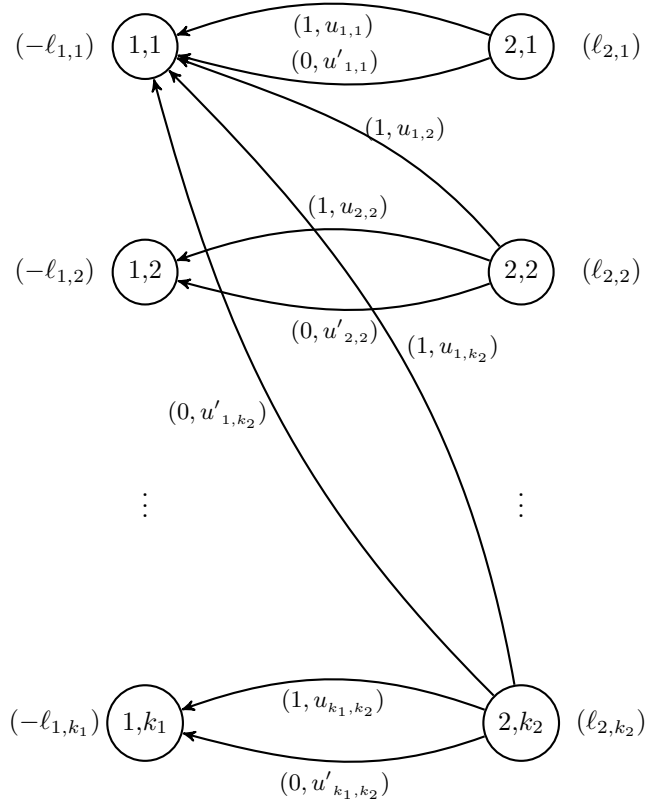
\begin{figure}[htb]
		\centering
		\begin{tikzpicture}	
		\begin{scope}[every node/.style={circle,thick,draw}]
		\node (11) at (0,0) {1,1};
		\node (12) at (0,-3) {1,2};
		\node (1k1) at (0,-9) {1,$k_1$};
		\node (21) at (5,0) {2,1};
		\node (22) at (5,-3) {2,2};
		\node (2k2) at (5,-9) {2,$k_2$};
		\end{scope}

		\begin{scope}[every node/.style={thick}]
		\node (1.) at (0,-6) {$\vdots$};
		\node (2.) at (5,-6) {$\vdots$};
		
		\node  at (-1.2,0) {$(-\ell_{1,1})$};
		\node  at (-1.2,-3) {$(-\ell_{1,2})$};
		\node  at (-1.2,-9) {$(-\ell_{1,k_1})$};
		\node  at (6.2,0) {$(\ell_{2,1})$};
		\node  at (6.2,-3) {$(\ell_{2,2})$};
		\node  at (6.2,-9) {$(\ell_{2,k_2})$};
		
		\end{scope}
		
		\begin{scope}[
		every edge/.style={draw=black, thick,font=\small,align=center, auto}]
		\path [->] (21) edge [out=160,in=20] node {$(1, u${\tiny $_{1,1}$}$)$} (11);
		\path [->] (22) edge [out=130,in=-20] node[right] {$(1, u${\tiny $_{1,2}$}$)$} (11);
		\path [->] (2k2) edge [out=100,in=-45] node[right] {$(1, u${\tiny $_{1,k_2}$}$)$} (11);
		\path [->] (22) edge  [out=160,in=20] node[pos=.45, above] {$(1, u${\tiny $_{2,2}$}$)$}  (12);
		\path [->] (2k2) edge [out=160,in=20] node {$(1, u${\tiny $_{k_1,k_2}$}$)$}  (1k1);
		
		\path [->] (21) edge [out=200,in=-15] node[above] {$(0, u'${\tiny $_{1,1}$}$)$} (11);
		\path [->] (22) edge [out=200,in=-15] node {$(0, u'${\tiny $_{2,2}$}$)$}  (12);
		\path [->] (2k2) edge [out=200,in=-15] node {$(0, u'${\tiny $_{k_1,k_2}$}$)$}  (1k1);
		\path [->] (2k2) edge [out=135,in=-75] node[left] {$(0, u'${\tiny $_{1,k_2}$}$)$} (11);
		
		\end{scope}
		
		\end{tikzpicture}
		\caption{Min-cost network flow graph after arc reversal.}
		\subcaption*{arc legend: (cost, upperbound); node legend: (supply)}
		\label{fig:mcnf2}
	\end{figure}

	The successive shortest path algorithm iteratively selects a node $s$ with excess supply (supply not yet sent to some demand node) and a node $t$ with unfulfilled demand and sends flow from $s$ to $t$ along a shortest path in the residual network \cite{Jewell58}, \cite{Iri60}, \cite{BG61}. The algorithm terminates when the flow satisfies all the flow balance constraints. Since at each iteration, the number of remaining units of supply to be sent is reduced by at least one unit, the number of iterations is bounded by the total amount of supply.  For the network in \cref{fig:mcnf2} the total supply is $n$.
	
	At each iteration, the shortest path can be solved with Dijkstra's algorithm of complexity $O(|A| + |V| \log |V|)$, where $|V|$ is number nodes and $|A|$ is number of arcs \cite{Tomizawa71}, \cite{EK72}. In our formulation, $|V|$ is $O\left(k_1+k_2\right)$, which is at most $O(n)$. 
	Since the number of nonempty sets $\lT_{1,i_1} \cap \lT_{2,i_2}$ is at most $\min\{n, k_1k_2 \}$, the number of unit-cost arcs is $O(\min\{n, k_1k_2 \})$.
	Since the number of nonempty sets $\lC_{1,i_1} \cap \lC_{2,i_2}$ is at most $\min\{n', k_1k_2 \}$, the number of zero-cost arcs is $O(\min\{n', k_1k_2 \})$. So the total number of arcs $|A|$ is $O(\min\{n+n', k_1k_2 \})$.
	
	Hence, the total running time of applying the successive shortest path algorithm on our formulation is $O(n\cdot (\min\{n+n', k_1k_2 \} + (k_1+k_2)\log (k_1+k_2))$.
\end{proof}

In contrast to the $2$-covariate \maxB problem which is polynomial time solvable, we show next that the $2$-covariate $\kappa$-\maxB problem is NP-hard when $\kappa\geq 3$.

\subsection{NP-hardness for the $2$-covariate $\kappa$-\maxB problem with $\kappa\geq3$}\label{sec:NPkappa}

We prove here that the $2$-covariate $\kappa$-\maxB problem is NP-hard for all constant values of $\kappa$ such that $\kappa\geq 3$. The proof reduces the {\em exact 3-cover} problem, which is NP-hard \cite{KARP72}\cite{GJ79}.
\\
{\em Exact 3-cover}: Given a collection $C$ of 3-element subsets (triplets) of a ground set $E$ with $|E|=3q$ for some integer $q$, is there a subcollection $C'\subseteq C$ where each element $e\in E$ appears in exactly one triplet of $C'$?

\begin{theorem}\label{thm:npP2k3} The $2$-covariate $\kappa$-\maxB problem is NP-hard for any constant $\kappa\geq3$.
\end{theorem}

\begin{proof}
	Given an instance of Exact 3-cover problem with ground set $E$ of size $3q$ and a collection of triplets $C$, we construct an instance of  $2$-covariate $\kappa$-\maxB problem for any constant integer $\kappa\geq 3$.
	
	First we define the levels of the two covariates.
	For covariate 1, we have a level $T$ for each triplet $T\in C$, and a level $e'$ for each element $e\in E$. So there are $|C|+|E|$ levels of covariate 1.
	For covariate 2, we have a level $e''$ for each element $e\in E$, and one additional level dentoed as $X$. So there are $|E|+1$ levels of covariate 2.
	
	Next, we construct the samples. For each sample, we represent it by the ordered pair $(a,b)$ where $a$ is the level of the first covariate and $b$ is the level of the second covariate. The treatment group contains a sample $(T,X)$ for each triplet $T\in C$ and a sample $(e',e'')$ for each element $e\in E$.
	Moreover, for each triplet $T\in C$, whose elements are denoted by $e_{t_1}, e_{t_2}, e_{t_3}$, we create three control samples $(T,e_{t_1}''), (T,e_{t_2}''), (T,e_{t_3}'')$, as well as $(\kappa-3)$ copies of control sample $(T,X)$. In addition, for each element $e\in E$ we have one control sample $(e',X)$ and $(\kappa-1)$ copies of control sample $(e',e'')$.
	The treatment group constructed is of size $|C|+|E|$ and the control group constructed is of size $\kappa|C|+\kappa|E|$. The two sizes are both polynomially bounded in the size of the Exact 3-cover instance so the reduction can be computed in polynomial time.
	
	Finally, we claim that the constructed  $2$-covariate $\kappa$-\maxB problem has a feasible solution with objective value of at least $4q$ if and only if the Exact 3-cover instance has a subcollection $C'\subseteq C$ such that each element $e\in E$ appears in exactly one triplet of $C'$.
	
	Let $C'$ be a subcollection such that each element $e\in E$ appears in exactly one triplet of $C'$, we derive a solution for the constructed problem as follows. In the treatment group, we choose $(T,X)$ for all $T\in C'$, and $(e',e'')$ for all $e\in E$. In the control group, for each $T\in C'$,  whose elements are denoted by $e_{t_1}, e_{t_2}, e_{t_3}$, we choose $(T,e_{t_1}''), (T,e_{t_2}''), (T,e_{t_3}'')$ and the $(\kappa-3)$ copies of $(T,X)$. Additionally, we also choose from the control group sample $(e',X)$ and $(\kappa-1)$ copies of $(e',e'')$ for each $e\in E$. That is, the selection of treatment samples is
	$$ S=\{  (T,X):\forall T\in C'  \}  \cup \{(e',e''): \forall e\in E \} $$
	and the selection of control samples is
	\begin{align*}
		S'=&\{ (T,e_{t_1}''), (T,e_{t_2}''), (T,e_{t_3}''):\forall T=(e_{t_1},e_{t_2},e_{t_3})\in C'  \}  \cup \{ (T,X)\sw):\forall T\in C', w=1,...,\kappa-3  \} \cup  \\
	&  \{(e',X): \forall e\in E \}  \cup \{(e',e'')\sw: \forall e\in E, w=1,...,\kappa-1 \} .
	\end{align*}
	To check the feasibility of this solution, first consider the levels of the first covariate. For each $T\in C$, if $T\in C'$ then there is exactly one treatment sample with level $T$ in $S$ and we choose exactly $\kappa$ such control samples for $S'$; and if $T\notin C'$, then we have not chosen any sample of this level for neither the treatment nor the control group. Furthermore, for each $e\in E$, we choose exactly one sample from the treatment group with covariate 1 level $e'$ and exactly $\kappa$ control samples with covariate 1 level $e'$.  Next, consider the levels of the second covariate.  We choose $|C'|$ number of level $X$ treatment samples and $(\kappa-3)|C'|+|E|$ number of level $X$ control samples. Note that the size of subcollection $C'$ must be $q$ as each element in $E$ appears in exactly one triplet of $C'$, so $(\kappa-3)|C'|+|E|=\kappa q=\kappa |C'|$.  For every $\bar{e}\in E$, we choose exactly one treatment sample with level $\bar{e}''$. There are $\kappa-1$ control samples of level $\bar{e}''$ under covariate 2 in the set $\{(e',\bar{e}'')\sw: \forall e\in E, w=1,...,\kappa-1 \}$. Since each element $\bar{e}$ appears in exactly one triple of $C'$, $\bar{e}''$ appears exactly once in the set $ \{ (T,e_{t_1}''), (T,e_{t_2}''), (T,e_{t_3}''):\forall T=(e_{t_1},e_{t_2},e_{t_3})\in C'  \}  $. So there are $\kappa$ control samples of level $\bar{e}''$ under covariate 2 in the selection $S'$. Therefore, this solution is feasible and the objective value of this solution is  $|S|=|C'|+|E|=4q$.

	On the other hand, if the constructed $2$-covariate $\kappa$-\maxB problem has a feasible solution $S,S'$ of objective value at least $4q$, we say that $T\in C$ is selected for subcollection $C'$ if the treatment sample $(T,X)$ is selected in $S$ for the constructed problem.  We will show that the subcollection of selected triplets is a feasible solution of the Exact 3-cover problem. Since there are only $3q$ samples in the treatment group that do not correspond to a triplets, the size of $C'$ must be at lease $q$. So the subcollection is feasible if we can establish that every pair of selected triplets is disjoint.  Assume by contradiction that there exist two selected subsets $T,T'\in C'$ that have a common element $e$.
	Due to $(S,S')$-$\kappa$-fine-balance, we know the number of control samples in each level under any covariate must be an integer multiple of $\kappa$.
	Since there is only one sample of level $\bar{e}''$ in the treatment group, in selection $S'$ there can be either $0$ or $\kappa$ samples of level $\bar{e}''$ under covariate 2.  Since both $(T,X)$ and $(T',X)$ are in the selection $S$, control samples $(T,\bar{e}'')$ and $(T',\bar{e}'')$ must be chosen in $S'$, otherwise the number of selected samples in covariate 1 levels $T$ and $T'$ will not satisfy $(S,S')$-$\kappa$-fine-balance. Thus, the number of samples in $S'$ with the second covariate being $\bar{e}''$ is at least $2$. So the number $|S'\cap \lC_{2,\bar{e}''}|$ must be $\kappa$, the number $|S\cap \lT_{2,\bar{e}''}|$ must be $1$,  and the number of sample $(\bar{e}',\bar{e}'')$ in selection $S'$ must be less than or equal to $\kappa-2$.  This implies that the number of control samples of level $\bar{e}'$ under covariate 1 can not be more than $\kappa-1$, which means, it must be zero. So we can derive that treatment sample $(\bar{e}',\bar{e}'')$ is not selected, which means the number of treatment samples in level $e''$ under covariate 2 is $0$, contradicts with $|S\cap \lT_{2,e''}|=1$.
\end{proof}

\section{The $\kappa$-balanced-matching ($\kappa$-BM) problem with $P\geq2$}\label{sec:kBM}

The $1$-covariate $\kappa$-BM problem is solvable in polynomial time \cite{Rosenbaum07}. However, we show here for the first time that even for any constant $\kappa$ the $\kappa$-BM problem is NP-hard for two or more covariates, $P\geq 2$. 
For the $2$-covariate BM problem and the $2$-covariate $\kappa$-BM problem, the complexity status when the numbers of levels of both covariates are constants is discussed in \cref{sec:FPT}  together with the other two families of problems. Here we consider an intermediate case where only one of the covariates has a constant number of levels.
We will show that the $2$-covariate BM problem and the $2$-covariate $\kappa$-BM problem when the second covariate has a constant number of levels can be solved efficiently if and only if the exact matching problem on bipartite graphs can be solved efficiently. 

\subsection{The hardness of the $P$-covariate BM problem and the $P$-covariate $\kappa$-BM problem for $P\geq2$} 
We will first present the integer programming formulation for the $2$-covariate $\kappa$-BM problem. Denote the distance between the $i$th treatment sample and the $j$th control sample by $\delta_{ij}$ for $i=1,...,n$ and $j=1,...,n'$.
The decision variable $x_{ij}$ is a binary variable which equals $1$ if and only if the $i$th treatment sample is matched to the $j$th control sample.
The integer programming formulation is given below.

\begin{subequations}\label{BM}
	\begin{align}
	& \min& \sum_{i=1}^n\sum_{j=1}^{n'} \delta_{ij} x_{ij} &&\label{objd} \\
	&\text{s.t.} &  \sum_{i=1}^n\sum_{j\in \lC_{1,q}} x_{ij} = \kappa \ell_{1,q}  &&   q=1,\ldots,k_1
	\label{cons:ab1} \\
	&&  \sum_{i=1}^n\sum_{j\in \lC_{2,q}} x_{ij} =\kappa \ell_{2,q}  &&   q=1,\ldots,k_2
	\label{cons:ab2} \\
	&& \sum_{j=1}^{n'} x_{ij} = \kappa &&  i=1,...,n \label{cons:sumt}  \\
	&& \sum_{i=1}^{n} x_{ij} \leq 1 &&  j=1,...,n' \label{cons:sumc}  \\
	&& x_{ij}\in\{0, 1\} &&  i=1,...,n , \quad j=1,\ldots,n'.
	\end{align}
\end{subequations}

The objective (\ref{objd}) is the total distance of all matched pairs.
Constraints (\ref{cons:ab1}) ensure that the number of matched level $q$ control samples is $\kappa$ times as large as the number of matched level $q$ treatment samples under covariate 1 for each level $q$. Constraints (\ref{cons:ab2}) ensure the $\kappa$-fine-balance requirement under covariate 2. Constraints (\ref{cons:sumt}) assign each treatment sample to $\kappa$ control samples, and constraints (\ref{cons:sumc}) specify that each control sample can not be matched to more than one treatment sample.

Next we prove that the $2$-covariate BM problem is NP-hard by reducing the {\em 3-SAT} problem, one of Karp's 21 NP-hard problems \cite{KARP72}\cite{GJ79}.
\\
{\em 3-SAT}: Given clauses $\cC_1, \cC_2, ..., \cC_m$, each consisting of 3 literals from the set $\{ v_1, v_2, ..., v_n\} \cup \{ \bar{v}_1, \bar{v}_2, ..., \bar{v}_n \}$. Is the conjunction of the given clauses satisfiable?

\begin{theorem}\label{thm:np2BM} The $2$-covariate BM problem is NP-hard.
\end{theorem}
	\begin{proof}
	Given an instance of 3-SAT problem with clauses $\cC_1, \cC_2, ..., \cC_m$ and variables $ v_1, v_2, ..., v_n$, we construct an instance of \MFB. We define two types of gadgets: the ``variable'' gadgets, one for each variable, and the ``clause'' gadgets, one for each clause.
	The distances between the treatment samples and the control samples are set to $0$ or $1$.
	We construct a \MFB\ such that the 3-SAT problem is satisfiable if and only if the optimal objective value is $0$. Moreover, the zero-distance matching implies the values of the 3-SAT variables in a truth assignment.
	
	First, we define the variable gadgets. Consider the $j$th variable $v_j$ that appears $p_j$ times as $v_j$ and $q_j$ times as $\bar{v}_j$.
	We have the following $2p_j+2q_j-2$ control samples in the $j$th variable gadget: $a_i\sj$ for $i=0,1,2,...,q_j-1$, $b_i\sj$ for $i=1,2,...,q_j-1$, $c_i\sj$ for $i=0,1,2,...,p_j-1$, and $d_i\sj$ for $i=1,2,...,p_j-1$. The $p_j+q_j-1$ treatment samples in this $v_j$ gadget are: $e_i\sj$ for $i=0,1,...,q_j-1$ and $f_i\sj$ for $i=1,...,p_j-1$. (We also call the $e_0\sj$ sample $f_0\sj$ for simplicity.)
	
	There are $p_j+q_j-1$ covariate 1 levels associated to the $v_j$ gadget, each one consists of three samples: a pair of control samples out of the $2p_j+2q_j-2$ control samples, and one treatment sample out of the $p_j+q_j-1$ treatment samples. The pairs of control samples that belong to the same levels are listed as follows:
	$\{a_i\sj,b_{i+1}\sj\}$ for $i=0,1,...,q_j-2$, $\{ c_i\sj,d_{i+1}\sj\}$ for $i=0,1,...,p_j-2$ and $\{ a_{q_j-1}\sj,c_{p_j-1}\sj \}$. Since each of these levels has one treatment sample, exactly one control sample out of those pairs need to be matched.
	
	The zero-distance pairs for the $v_j$ gadget are $[e_i\sj,a_i\sj]$ for $i=0,1,...,q_j-1$, $[e_i\sj,b_i\sj]$ for $i=1,2,...,q-1$, $[f_i\sj,c_i\sj]$ for $i=0,1,...,p_j-1$ and $[f_i\sj,d_i\sj]$ for $i=1,2,...,p_j-1$. Each treatment sample has two potential matches with zero distance, and each control samples has one potential matches with zero distance.
	Note that sample $e_0\sj=f_0\sj$ must be matched to either $a_0\sj$ or $c_0\sj$ in a zero-distance matching. So there are two possible zero-distance matchings:
	\begin{itemize}
		\item Case 1 ($e_0\sj$ is matched to $a_0\sj$): $e_i\sj$ is matched to $a_i\sj$ for $i=0,1,2,...,q_j-1$, $f_i\sj$ is matched to $d_i\sj$ for $i=1,...,p_j-1$; samples $\{b_i\sj:i=1,2,...,q_j-1\}$ and $\{c_i\sj:i=0,1,2,...,p_j-1\}$ are unmatched;
		\item Case 2 ($e_0\sj=f_0\sj$ is matched to $c_0\sj$): $f_i\sj$ is matched to $c_i\sj$ for $i=0,1,2,...,p_j-1$, $e_i\sj$ is matched to $b_i\sj$ for $i=1,...,q_j-1$; samples $\{a_i\sj:i=0,2,...,q_j-1\}$ and $\{d_i\sj:i=1,1,2,...,p_j-1\}$ are unmatched;
	\end{itemize}
	Consider the first case in which $e_0\sj$ is matched to $a_0\sj$. For $i=0,...,q_j-2$, if treatment sample $e_i\sj$ is matched to $a_i\sj$, as only one control sample from the same level can be matched, $b_{i+1}\sj$ must not be matched. Then as the two only zero-distance matches for $e_{i+1}\sj$ are $a_{i+1}\sj$ and $b_{i+1}\sj$, we infer that $e_{i+1}\sj$ must be matched to $a_{i+1}\sj$. By such induction, we know all samples in $\{a_i\sj:i=0,1,2,...,q_j-1\}$ are matched, and samples in $\{b_i\sj:i=1,2,...,q_j-1\}$ are unmatched.
	Since the only zero-distance match of $c_0\sj$ is $e_0\sj$, $c_0\sj $is unmatched in Case 1. For $i=0,...,p_j-2$, if sample $c_i\sj$ is unmatched,
	as one control sample need to be matched in each level, we can infer that sample $d_{i+1}\sj$ is matched to its only zero-distance pair $f_{i+1}\sj$. Therefore, $c_{i+1}\sj$ can not be matched as its only zero-distance pair is taken. By such induction, we know samples in $\{c_i\sj:i=0,1,2,...,p_j-1\}$ are unmatched, and all samples in $\{d_i\sj:i=1,2,...,p_j-1\}$ are matched.
	With similar arguments, in the second case in which  $e_0\sj=f_0\sj$ is matched to $c_0\sj$, all samples in $\{b_i\sj:i=1,2,...,q_j-1\}\cup \{c_i\sj:i=0,1,2,...,p_j-1\}$ are matched, and samples in $\{a_i\sj:i=0,2,...,q_j-1\}\cup\{d_i\sj:i=1,1,2,...,p_j-1\}$ are unmatched.
	In Case 1, we say that we assign variable $v_j$ the value FALSE, and in Case 2 we say that we assign variable $v_j$ the value TRUE.
	
	Next, consider the clause gadgets.
	For clause $\cC_w$ that consists of variables $v_{w_1},v_{w_2},v_{w_3}$,
	we pick three samples without replacement from the variable gadgets that correspond to $v_{w_1},v_{w_2},v_{w_3}$ as follows. If the variable appears as literal $v_{j(w)}$ in the clause, we pick one of the samples of the type $c_i\sjw$ from the $v_{j(w)}$ gadget; if  the variable appears as literal $\bar{v}_{j(w)}$ we pick one of the samples of the type $a_i\sjw$ from the $v_{j(w)}$ gadget.
	Since for each occurrence of  literal $v_j$ there is a type $c_i\sj$ sample, and for each occurrence of  literal $\bar{v}_w$ there is a type $a_i\sj$ sample in the $v_j$ gadget, we can ensure that there are enough samples to be selected from without replacement.
	The clause gadget is then augmented by eight new samples used as ``garbage collectors'': three treatment samples $g_1\sw, g_2\sw, g_3\sw$ and five control samples $h_1\sw, h_2\sw, h_3\sw, h_1\dw, h_2\dw$. The fifteen pairs of distances between the three treatment samples and the five control samples are set to zero.
	
	We introduce a new covariate 1 level that includes all of the eight new samples, so that it is disjoint from the other levels created for the variable gadgets. Since there are three treatment samples in this level, three of the control samples $h_1\sw, h_2\sw, h_3\sw, h_1\dw, h_2\dw$ have to be matched.
	
	We introduce a new covariate 2 level which consists of five control samples: $h_1\dw, h_2\dw$ and the three previously selected samples of type $a_i$ or $c_i$ in the $v_{w_1},v_{w_2},v_{w_3}$ gadgets. We also set the three treatment samples $g_1\sw, g_2\sw, g_3\sw$ to be in this covariate 2 level.  Therefore, three out of the five control samples of this level must be matched. Observe that any three control samples in this level must contain at least one sample of type $a_i$ or $c_i$ from a variable gadget. That means, the clause is satisfied by our variable assignment rule above: suppose the matched literal appears as $v_{j(w)}$ and the sample selected for this clause gadget is $c_i\sjw$, sample $c_i\sjw$ is matched so variable $v_{j(w)}$ must be assigned the value TRUE from our previous discussion; and suppose the matched literal appears as $\bar{v}_{j(w)}$ and the sample selected for this clause gadget is $a_i\sjw$, sample $a_i\sjw$ is matched so  variable $v_{j(w)}$ must be assigned the value FALSE.
	
	For all remaining samples that have not been assigned to a covariate 2 level in the above discussion, i.e. the treatment samples in the variable gadgets and the control samples of type $h_1\sw, h_2\sw, h_3\sw$,  we create a ``dump'' covariate 2 level that includes all of them.
	
	With the construction above, a zero value matching solution implies a  truth assignment of the 3-SAT problem. The values of the 3-SAT variables are determined by which of the two zero-distance matchings is in each variable gadget. For each clause gadget, at least one of the type $a_i$ or $c_i$ sample is matched, which implies the clause is satisfied by our variable assignment rule.
	
	On the other hand, given a truth assignment of the 3-SAT problem, there is a zero-distance matching for the constructed problem. First match the samples in the variable gadgets as follows: if the variable takes value TRUE, then match the samples as described above in Case 2; if it takes value FALSE, then match the samples as described in Case 1. By doing so, the number of matched control samples under the covariate 1 level associated to each variable gadget, equals the number of treatment samples in that level.
	Next, for each clause gadget, if there is one, respectively two, respectively three satisfied literals, then match two, respectively one, respectively zero out of $\{h'_1, h'_2\}$ to $g_1, g_2$. That will match exactly three samples of the corresponding covariate 2 level as required.
	Finally, for $w=1,...,m$, the remaining unmatched treatment samples in clause $\cC_w$ are matched to control samples out of  $\{h_1\sw, h_2\sw,h_3\sw\}$. Since the three treatment samples $g_1\sw, g_2\sw, g_3\sw$ are matched to the control samples in $\{h_1\sw, h_2\sw,h_3\sw,h_1\dw, h_2\dw\}$, the numbers of matched samples in the associated covariate 1 level are three for both the treatment and control sides. For either the treatment or control side, the number of matched samples in the dump level under covariate 2 is the total number of matched samples minus the number of matched samples in other covariate 2 levels. As we have shown that the number of matched samples in the covariate 2 level associated with each clause gadget are three for both sides, the numbers of matched samples in the dump level also equal to each other for the two groups. So this is a zero-distance matching for the constructed \MFB.

\end{proof}

This proof of \cref{thm:np2BM} can be extended to the  $2$-covariate $\kappa$-BM problem for any constant $\kappa$.

\begin{corollary}\label{cor:np2kBM} The $2$-covariate $\kappa$-BM problem is NP-hard for any constant $\kappa$.
\end{corollary}
\begin{proof}
	We prove this corollary also by reduction from the 3-SAT problem with clauses $\cC_1, \cC_2, ..., \cC_m$ and variables $ v_1, v_2, ..., v_n$.
	
	We first construct the same covariates levels, treatment and control samples which are the same as the  $2$-covariate BM instance in the proof of \cref{thm:np2BM}. Let $t$ denote the number of treatment samples constructed. The distance between each treatment and each control is modified: we change all distances $0$ to $1$ and all distances $1$ to $M$ for $M$ being a large constant which is greater than $t$. According to the proof of \cref{thm:np2BM},  the 3-SAT problem is satisfiable if and only if the optimal objective value of the $2$-covariate BM problem with this modified distance  is $t$.
	
	Next, we add more samples to the control group. For each treatment sample constructed, we add $\kappa-1$ number of control samples that take the same level values as the treatment sample for both covariates. The $\kappa-1$ distances between the treatment sample and the new $\kappa-1$ control copies of it are set to $0$. All remaining undefined distances are set to $M$. We claim that the 3-SAT problem is satisfiable if and only if the optimal objective value of the $2$-covariate $\kappa$-BM problem on this new instance is $t$.
	
	If the 3-SAT problem is satisfiable, we can find an assignment for the BM problem on the constructed instance with distance $t$ as in the proof of \cref{thm:np2BM}. By assigning additionally each treatment sample  to the corresponding $\kappa-1$ new control copies, we have a solution for the $\kappa$-BM problem on the constructed instance. Furthermore, this is also the optimal solution for the $\kappa$-BM problem as there are only $(\kappa-1)\cdot t$ zero-distance pairs.
	
	On the other hand, if the optimal solution to the $2$-covariate $\kappa$-BM problem has a total distance of $t$, then all the $(\kappa-1)\cdot t$ zero-distance pairs must be matched. In addition, there must be $t$ matched pair with distance $1$. From the arguments in the proof of \cref{thm:np2BM}, we can derive that there is a truth assignment for the 3-SAT problem.
\end{proof}

We can further derive that the $\kappa$-BM problem is also NP-hard for more than two covariates.

\begin{corollary}\label{cor:np2BM}
The $P$-covariate $\kappa$-BM problem is NP-hard for every value of $P$ when $P\geq2$ for any constant $\kappa$.
\end{corollary}
\begin{proof}
	
	For any $2$-covariate $\kappa$-BM problem, and any $P\geq3$, we can construct an equivalent $P$-covariate $\kappa$-BM problem  by adding $P-2$ covariates for each sample in the $2$-covariate $\kappa$-BM problem instance and set the value of the $p$th covariate to be the same for all samples, for each $p=3,...,P$.
	
	Therefore, the NP-hardness of $2$-covariate $\kappa$-BM problem implies that the $P$-covariate $\kappa$-BM problem is NP-hard as long as $P\geq2$.
\end{proof}

\subsection{The special case of $2$-covariate BM and $2$-covariate $\kappa$-BM problems where one covariate has a constant number of levels}

Let BM' be the special case of $2$-covariate BM problem where the second covariate has a constant number of levels while the first covariate has no restriction on the number of levels. In Section \ref{sec:FPT} we will establish that if both covariates have constant number of levels then the $2$-covariate BM problem is polynomial time solvable.  We show here that the complexity status of the $2$-covariate problem in which only  {\em one} covariate has a constant number of levels is linked to the complexity status of the {\em exact matching problem} and its weighted version denoted as {\em weighted exact matching}.  In order to present this connection we assume that the distance matrix is integral and all distances are given in unary, that is, there is a polynomial $\pi$ of the input encoding length where $\delta_{ij} \leq \pi$ for all $i,j$.

The exact matching in bipartite graph problem is defined as follows.  The input is an integer number $k$ together with a bipartite graph $G=(V_1\cup V_2,E)$ with $|V_1|=|V_2|=q$, and its edge set $E$ is partitioned into $E_b \cup E_r$ where $E_b$ is the set of blue edges and $E_r$ is the set of red edges. The exact matching problem is to find a perfect matching that has exactly $k$ blue edges (and all other $q-k$ edges are red).  The complexity status of the exact matching problem is as follows.  While \cite{MVV} showed that there is a randomized polynomial time algorithm for the problem, the existence of a deterministic polynomial time algorithm is still an important open problem.  

The weighted exact matching problem is defined as follows.  The input is a bipartite graph $G=(V,E)$ together with non-negative integral distances $\delta_e$ for all $e\in E$, where there is a polynomial $\pi$ of $|V|+|E|$ such that $\delta_e \leq \pi$ for all $e\in E$.  We are also given a target value $K$.  The goal is to find a perfect matching of total distance exactly $K$. Note that the weighted exact matching problem is a generalization of the exact matching problem since the later problem can be interpreted as the weighted exact matching problem where the weight of a blue edge is $1$ and the weight of a red edge is $0$.  Thus, a polynomial time algorithm for the weighted exact matching problem gives a polynomial time algorithm for the exact matching problem. On the other hand, it is known that a polynomial time algorithm for the exact matching problem gives a polynomial time algorithm for the weighted exact matching problem (see proposition 1 in \cite{PY}).  If the algorithm for the exact matching is deterministic (randomized), then the algorithm for the weighted exact matching is deterministic (randomized, respectively) as well \cite{PY}.  Therefore, the complexity of the weighted problem has the same status as the one of the exact matching problem.  Namely, the result of \cite{MVV} gives a randomized polynomial time algorithm for  the weighted exact matching problem, while the existence of a deterministic polynomial time algorithm for this problem will result in a deterministic polynomial time algorithm for the exact matching in bipartite graphs problem.

We show the following connections between the exact matching problem (or the weighted exact matching) and problem BM'.

\begin{theorem}
If there is a deterministic (or randomized) polynomial time algorithm for BM' then there is a deterministic (or randomized, respectively) polynomial time algorithm for exact matching in bipartite graphs.
\end{theorem}
\begin{proof}
Assume that there is a polynomial time algorithm ALG for BM' and we will establish the existence of a polynomial time algorithm for the exact matching problem. Given an input to the exact matching problem with $2q$ nodes ($q$ nodes on each side of the bipartite graph), we denote the bipartition of the graph by $V_1\cup V_2$ and the partition of the edge set into $E_b \cup E_r$, and we let $k$ be the number of the required blue edges in the matching.
We define the following input for BM'.  We will associate samples with nodes so the control group consists of nodes and the treatment group also consists of nodes. For every node $v\in V_1$, we have two nodes in the control corresponding to $v$: a red node $v_r$ and a blue node $v_b$. All blue edges in $E_b$ that were incident to $v$ in the input to the exact matching problem are now incident to $v_b$, and all red edges that were incident to $v$ are now incident to $v_r$. These edges corresponding to original edges in the input for the exact matching instance have zero distance, while all other distances are set to $1$. The nodes in $V_2$ (of the original input graph to the exact matching) are the treatment nodes, so the distances we defined represent the distances between a treatment node and a control node. 

The levels of the first covariate are defined such that every pair $[v_r,v_b]$ for $v\in V_1$ defines one level of the first covariate, and we have one treatment node in each such level.  Observe that the number of levels of the first covariate is $q$, and we have $q$ nodes in the treatment group, so this assignment of levels of the first covariate is feasible.
Next, consider the second covariate. We will have two levels of the second covariate corresponding to blue and red. The red level of the control is the set of all red nodes, and the blue level of the control is the set of all blue nodes. The second level of the treatment are defined so that there will be exactly $k$ nodes of the treatment in the blue level and the remaining $q-k$ nodes in the treatment are in the red level.  

We would apply algorithm ALG on the BM' instance and check if the output cost is zero or strictly positive. In any feasible solution of the BM' defined, exactly one of two control nodes $[v_r,v_b]$ is matched for each $v\in V_1$, as there is exactly one treatment node in each level of the first covariate. And since we need to match $k$ blue control nodes and $q-k$ red control nodes, a zero distance matching of the BM' represents a set of edges in the exact matching instance that is a perfect matching consisting of $k$ blue edges and $q-k$ red edges.
 
Observe that this construction is of deterministic polynomial time. Thus the algorithm that constructs the input to BM' and apply ALG on that input is a deterministic (randomized) polynomial time algorithm for the exact matching problem if ALG is a deterministic (randomized, respectively) polynomial time algorithm for BM'.  
\end{proof}

We next consider the other direction.
\begin{theorem}
If there is a deterministic (or randomized) polynomial time algorithm for weighted exact matching in bipartite graphs then there is a deterministic (or randomized, respectively) polynomial time algorithm for BM'.
\end{theorem}
\begin{proof}
Assume that there is a polynomial time algorithm for the weighted exact matching problem in bipartite graph, we will establish the existence of a polynomial time algorithm for BM'. Here we are going to use the fact that the maximum distance is at most $\pi$ and without loss of generality, we assume that $n' \leq \pi$. We set $\eps=1/(2n^2 \pi+1)$, and we define a multi-objective optimization problem. As a step in the algorithm for solving BM', we will find a $(1+\eps)$-approximated Pareto set of perfect matchings in the following graph with the following multi-criteria objective.

We consider a complete bipartite graph where one side consists of the control group, namely one node for each control sample, the other side of the graph correspond to the treatment group together with some additional nodes.  If we have in the instance of BM', a control group of size $n'$ and a treatment group of size $n$, then we will have exactly $n'-n$ additional nodes.  The graph that we consider is the complete bipartite graph with $n'$ nodes on each side.  In this graph a feasible solution is a perfect matching of all nodes. 

Next, we define the $1+k_2$ objectives, for constant $k_2$ being the number of levels of the second covariate. The first objective corresponds to total distance, and one of the remaining objectives for each level of the second covariate. These objectives are sums of cost for edges in the perfect matching with different coefficients. We denote these cost coefficients as a vector for each edge so the first coefficient of this cost is the coefficient of the cost function of the first objective.
In order to define the first cost coefficient, we assign a covariate 1 level for the additional nodes as follows.  For each level $p$ of the first covariate, we assign exactly $\ell'_{1,p}-\ell_{1,p}$ additional nodes to have level $p$.  Note that without loss of generality we have $\ell'_{1,p}\geq \ell_{1,p}$ for all $p$ (as otherwise BM' is infeasible), and thus our definition of levels of the first covariate for the additional nodes, indeed define such level for every additional node. 
Consider an edge $[u,v]$ between an additional node $u$ and a control sample node $v$. The first cost coefficient of edge $[u,v]$ is $0$ if the two nodes share the same level of covariate 1, otherwise it is set to $2n\pi$.  The other $k_2$ cost coefficients of such an edge $[u,v]$ are set to $0$.
Consider next an edge $[u',v]$ where $u'$ is a treatment node (and not an additional node) and assume that the second covariate level of $v$ is $p$.  The first cost coefficient is the distance between samples $u'$ and $v$ in the BM' instance (that is, $\delta_{u',v}$). The $(p+1)$th cost coefficient is set to $1$ and the other $k_2-1$ coefficients are set to $0$. In this way, the $(p+1)$th objective is to minimize the number of edges adjacent to  level-$p$ control nodes for covariate 2. 

Observe that every feasible solution, namely, every perfect matching in this bipartite graph has the property that the sum of the costs according to all objectives excluding the first one is exactly $n$.  Furthermore, observe that in a Pareto set of this multi-criteria optimization problem, there is exactly one point where for every $p=1,2,\ldots ,k_2$, the total cost of the matched edges according to the $(p+1)$-th objective is exactly the number of treatment samples of the $p$-th level of the second covariate.  We will refer to this point as the {\em candidate feasible point of the Pareto set}.

Now, we use the result of Papadimitriou and Yannakakis \cite{PY00} to conclude that the algorithm for weighted exact matching gives the required algorithm for approximating the Pareto set and finding the $(1+\eps)$-approximated Pareto set (see theorem 4 and corollary 5 in \cite{PY00}).  We can use the results of \cite{PY00} since the maximum cost coefficient of an edge in our instance is $2n\pi$, that is polynomially bounded in the input encoding length, and the number of objectives is a constant.

We next argue that the $(1+\eps)$-approximated Pareto set is actually the Pareto set of the multi-criteria problem.  The cost vectors of the edges are always integral and have a maximum coefficient of at most $2n\pi$, and thus for every objective the cost of the matching is at most $2n^2\pi$.  Therefore, if we approximate this objective with an approximation ratio of $1+\eps$ then by our choice of $\eps$ we get the optimal value of this objective.  
By the last claim we conclude that the candidate feasible point of the Pareto set is one of the solutions that appear in this Pareto set and we consider this specific solution.

We delete from the candidate solution the edges (in the matching) that are adjacent to the additional nodes. By the notion of the candidate solution, the fine balance constraints of the second covariate are satisfied. So if the resulting matching is not a feasible solution to the BM' instance, it means that there is at least one additional node that used to be matched (in the candidate solution) to a control sample of a different covariate 1 level.  This is so as otherwise, for every level of the first covariate, the number of selected control samples of this level is the same as the number of treatment samples of this level.  Consequently, if the resulting matching is infeasible for BM', then the first objective value of the candidate solution is at least $2n\pi$, which implies that the BM' problem is infeasible as we show next. If BM' has a feasible solution, then we can create an alternative candidate solution by adding to this solution for BM' a zero distance matching of the additional nodes. This alternative solution has a first objective value that is less than or equal to $n\pi$, and all other objectives values are the same as the candidate solution. So according to the $(1+\eps)$ Pareto optimality, the resulting matching must be feasible for the BM' problem.  

In summary, we consider the candidate solution for BM' obtained from the candidate feasible point of the Pareto set after deleting the edges adjacent to the additional nodes. We check if the candidate solution for BM' satisfies the fine balance constraints.  If it does, then this is the output of the algorithm for BM', and if it does not, then the BM' instance is infeasible.

Furthermore, if we use the existing algorithm of \cite{MVV} to solve the weighted exact matching then the resulting algorithm will be randomized polynomial time algorithm whereas if we will use it with a deterministic polynomial time algorithm then the resulting algorithm for BM' is also a deterministic polynomial time algorithm.
\end{proof}

Next we consider the $\kappa$-BM' that is the special case of $2$-covariate $\kappa$-BM where the second covariate has a constant number of levels, and once again we assume that the distance matrix is integral and the maximum distance is upper bounded by a polynomial $\pi$ of the input encoding length.  We show that $\kappa$-BM' has the same complexity status as BM' (for all $\kappa\geq 2$).  That is, we establish the following result.

\begin{theorem}
There is a polynomial time algorithm for BM' if and only if there is a polynomial time algorithm for $\kappa$-BM'.
\end{theorem} 
\begin{proof}
Assume that there is a polynomial time algorithm ALG for BM'.  Consider an instance of the $\kappa$-BM' problem, and replace every treatment sample by $\kappa$ copies of it with the same pair of levels as the original element of the treatment group.  We define the distance matrix as follows.  The distance between a treatment sample $x$ that is a copy of the original treatment sample $x'$ (of the instance for the $\kappa$-BM') and a control sample $y$, is now defined as the distance between $x'$ and $y$.  The resulting treatment group and control group is the instance for BM', and we apply ALG on that instance.  A feasible solution for this BM' instance gives a feasible solution for the $\kappa$-BM' instance (simply by matching a treatment sample to a control sample if one of the copies of the treatment sample was matched to that control sample) and of the same total distance.  Similarly, a feasible solution to the original $\kappa$-BM' instance gives a feasible solution for the BM' instance of the same cost by matching the set of $\kappa$ matched control sample to the unique treatment sample to the $\kappa$ copies of this treatment sample in the BM' instance (with one control sample matched to each copy).  Thus, we get a polynomial time algorithm for $\kappa$-BM' problem.

Consider the other direction.  Assume that we are given a polynomial time algorithm ALG for $\kappa$-BM' and we establish the existence of a polynomial time algorithm for BM'.  Consider an instance of the BM' problem.  First, we add $(n+1)\pi$ for every component of the distance matrix.  This modification of the distance matrix ensures that every feasible solution for BM' has a cost that is at least $(n^2+n)\pi$ and at most $(n^2 +2n)\pi$.  Next, for every treatment sample $x$ we add another $\kappa-1$ control samples with the same levels as the ones of $x$ whose distance from $x$ is $0$ and from other treatment samples the distance is $(n^2+2n)\cdot \pi+1$. These $\kappa-1$ control samples for each treatment sample are called dummy control samples. This is the instance of $\kappa$-BM' on which we apply ALG.  The output has cost  $B\leq (n^2+2n) \cdot \pi$ if and only if the solution we obtain by deleting all the dummy control sample  is a feasible solution for BM' of cost $B$.  To prove the last claim note that the solution for the $\kappa$-BM' instance cannot match dummy control samples to treatment samples if their distance is not zero because the distance of such match is larger than $B$ and all distances are non-negative.  Furthermore, the $\kappa-1$ zero distances for each treatment sample must be in the selected control samples as otherwise, the total distance would be at least $(n+1)^2 \cdot\pi > (n^2+2n)\cdot \pi$.  Thus, by deleting the dummy control sample we get a feasible solution for the BM' problem (after the modification of the distance matrix) of the same cost.  Similarly, if we take an optimal solution for the BM' problem (before or after the modification of the distances) and add to it the dummy control samples that are matched using the zero-distances then we get an optimal solution for the $\kappa$-BM'.  Therefore, by applying ALG on the last $\kappa$-BM' problem we get a polynomial time algorithm for solving the original BM' instance.      
\end{proof}

\section{The maximum selection $\kappa$-fine-balance matching ($\kappa$-MSBM) problem}\label{sec:MSBM}

Since any $\kappa$-BM problem can be solved as a $\kappa$-MSBM problem with the same $\kappa$ and the same number of covariates, \cref{cor:np2BM} implies that the $P$-covariate $\kappa$-MSBM problem is also NP-hard for $P\geq 2$ and any constant $\kappa$. And in \cref{sec:prelim} we show that the $1$-covariate MSBM problem can be solved as an MCNF problem in polynomial time. We are going to show next that the $1$-covariate $\kappa$-MSBM problem is NP-hard even for any constant $\kappa\geq 3$ by reduction from the Exact-3-cover problem (see \cref{sec:NPkappa}) .

\begin{theorem}\label{thm:NPkMSBM}
	For any constant value of $\kappa$ such that $\kappa\geq 3$, the $1$-covariate $\kappa$-MSBM problem is NP-hard even with only one level.
\end{theorem}
\begin{proof}
	Given an instance of Exact-3-cover, namely a collection $C$ of 3-element subsets (triplets) of a ground set $E$ with $|E|=3q$ for some integer $q$.  We will define an instance of $1$-covariate $\kappa$-MSBM with only one level for any constant $\kappa$ such that $\kappa\geq 3$. In the constructed instance, we have one treatment sample for every triplet in $C$, and we have one control sample for every element $e\in E$.  In addition we have $(\kappa-3)\cdot q$ {\em dummy} control samples. For each triplet $T\in C$, the distance of the corresponding treatment sample to a control sample is defined as follows: it is zero if the control sample is one of the dummy samples or if the control sample is an element of the triplet $T$.  All other distances (that are still undefined) are set to one. Observe that a feasible solution for the $\kappa$-MSBM instance selects all control group and selects exactly $q$ treatment samples. We claim that in this instance, an optimal solution for $\kappa$-MSBM has total distance of zero, if and only if the Exact-3-cover instance is a YES instance.
	
	To see the last claim assume first that there is a subcollection of triplets $C'\subseteq C$ such that every element of $E$ appears in exactly one triplet in $C'$.  Then we construct a zero-distance solution for the $\kappa$-MSBM instance as follows.  We select the samples $C'$ of the treatment group and each such selected sample $T\in C'$ is matched to the samples of the control consisting of the three elements samples in $T$ together with $\kappa-3$ of the dummy control samples.  Since $C'$ has $q$ triplets, we have sufficient number of additional control samples.  Furthermore, since every element in $E$ appears only once in triplets of $C'$, we conclude that its control sample is matched to exactly one selected treatment sample.
	
	On the other hand, assume that there is a zero-distance solution for the $\kappa$-MSBM instance. Then, this solution selects exactly $q$ treatment samples corresponding to $q$ triplets.  Denote by $C'\subseteq C$ this subcollection of $q$ triplets.  Note that if there is a selected treatment sample that is matched to at least $\kappa-2$ dummy control samples, then there exists a selected treatment sample that is matched to at most $\kappa-4$ dummy control samples, and thus it is matched to at least four control samples that correspond to elements of $E$, then at least one of those matches has distance one.  This contradicts the assumption that the cost of the $\kappa$-MSBM solution is zero.  Therefore, every selected treatment sample is matched to exactly $\kappa-3$ additional control samples and three control samples that are corresponding to its elements.  Since every control sample corresponding to an element is matched exactly once, we conclude that every element of $E$ appears in exactly one triplet of $C'$, so the Exact-3-cover instance is a YES instance.
\end{proof}

For the $1$-covariate $\kappa$-MSBM problem where $\kappa=2$, the complexity status remains open.

\section{Fixed-parameter tractable algorithms}\label{sec:FPT}

In this section, we consider the special cases of the $\kappa$-FBS, $\kappa$-BM, and $\kappa$-MSBM problems where all covariates have a small number of levels.

Let $K=\prod_{i=1}^P k_i$ be the number of level-intersections. Observe that if the number of covariates is constant and all covariates have constant number of levels, then $K$ is a constant. We note that the problems $\kappa$-FBS, $\kappa$-BM, and MSBM can be solved in fixed-parameter tractable (FPT) time with parameter $K$.
In order to state these results, we say that a problem is fixed-parameterized complexity with parameter $K$ and denote it by $FPT(K)$ if it has an algorithm whose time complexity is upper bounded by a function of the form $f(K) \cdot \mbox{{\tt poly}}$ where $f(K)$ is some computable function of the parameter $K$, and $\mbox{{\tt poly}}$ is some polynomial of the input binary encoding length. We also say that an algorithm runs in $FPT(K)$ time and mean that its time complexity can be upper bounded by a function of the form $f(K) \cdot \mbox{{\tt poly}}$ where $f(K)$ is some computable function of the parameter $K$, and $\mbox{{\tt poly}}$ is some polynomial of the input binary encoding length. Here we show that these problems, namely  $\kappa$-FBS, $\kappa$-BM problems for all $\kappa$, and MSBM problem are $FPT(K)$. Similar results for $\kappa$-MSBM where $\kappa\geq 3$ cannot be obtained unless $P=NP$ as shown in \cref{thm:NPkMSBM}. The complexity status of the 2-MSBM problem with constant $K$ is open.

Our proof for the $FPT(K)$ results uses the existence of fast algorithms for solving integer programming in fixed dimension and for solving mixed-integer linear programs if the number of integral variables is fixed.   Lenstra \cite{Lenstra83} (see also \cite{Kannan83} for an improved time complexity of these algorithms) showed that  the integer linear programming problem with a fixed number of variables is polynomially solvable, and he also showed that a mixed-integer linear program with a fixed number of integer variables can be solved in polynomial time.  In fact, these algorithms runs in $FPT$ time with parameter being the number of integral variables. Therefore, to prove our results we show either an integer programming (IP) formulation with number of decision variables $O(K)$ or a mixed-integer linear program (MILP) with $O(K)$ integer variables such that solving this MILP to optimality ensures that the resulting solution is integral and solves the corresponding problem.

\subsection{The $\kappa$-FBS problem}
First consider the $\kappa$-FBS problem. For this problem we use an integer program with dimension $O(K)$ that is based on \IP. Let $u_{i_1,i_2,\ldots ,i_P}=|\lT_{1,i_1} \cap \lT_{2,i_2}\cap ... \cap  \lT_{P,i_p}|$ and $u'_{i_1,i_2,\ldots ,i_P}=|\lC_{1,i_1} \cap \lC_{2,i_2}\cap ... \cap  \lC_{P,i_p}|$ for $ i_p=1,...,k_p,\ p=1,...,P$. The decision variables are:
\\ \indent $x_{i_1,i_2,\ldots ,i_P}$: the number of  treatment samples selected from the $(i_1,i_2,\ldots ,i_P)$ level intersection $\lT_{1,i_1} \cap \lT_{2,i_2}\cap ... \cap  \lT_{P,i_p}$, for $ i_p=1,...,k_p,\ p=1,...,P$;
\\ \indent  $x'_{i_1,i_2,\ldots ,i_P}$: the number of  control samples selected from the $(i_1,i_2,\ldots ,i_P)$ level intersection $\lC_{1,i_1} \cap \lC_{2,i_2}\cap ... \cap  \lC_{P,i_p}$, for $ i_p=1,...,k_p,\ p=1,...,P$. \\
The integer programming formulation is:

\begin{subequations}\label{PFBS}
	\begin{align}
	& \max &   \sum_{i_1=1}^{k_1} \sum_{i_2=1}^{k_2}\cdots \sum_{i_P=1}^{k_P} x_{i_1,i_2,\ldots ,i_P}  &\label{Pobj} \\
	&\text{s.t.} &  \kappa\cdot \sum_{i_1=1}^{k_1}...\sum_{i_{p-1}=1}^{k_{p-1}}\sum_{i_{p+1}=1}^{k_{p+1}}...\sum_{i_{P}=1}^{k_{P}} x_{i_1,i_2,\ldots ,i_P} = \sum_{i_1=1}^{k_1}...\sum_{i_{p-1}=1}^{k_{p-1}}\sum_{i_{p+1}=1}^{k_{p+1}}...\sum_{i_{P}=1}^{k_{P}}  x'_{i_1,i_2,\ldots ,i_P}    && \nonumber \\
	&&p=1,\ldots ,P \ \  i_p=1,...,k_p&&\label{Pcons:b1} \\
	&& 0\leq  x_{i_1,i_2,\ldots ,i_P} \leq u_{i_1,i_2,\ldots ,i_P}\ \ \ \  p=1,...,P,\ \  i_p=1,...,k_p&& \label{Pcons:t}\\
	&& 0\leq x'_{i_1,i_2,\ldots ,i_P} \leq u'_{i_1,i_2,\ldots ,i_P} \ \ \ \ p=1,...,P,\ \  i_p=1,...,k_p &&  \label{Pcons:c}\\
	&& x_{i_1,i_2,\ldots ,i_P}, x'_{i_1,i_2,\ldots ,i_P}\text{ integers} \ \ \ \  p=1,...,P,\ \  i_p=1,...,k_p &&  \label{Pcons:int}
	\end{align}.
\end{subequations}

Note that this integer programming formulation has $2K$ decision variables and $O(K)$ constraints, and thus the algorithm that constructs it and solves it to optimality runs in $FPT(K)$ time.  The optimal solution for this integer program encodes the optimal solution for $\kappa$-FBS similarly to the proof of theorem \ref{thm:unique}

\subsection{The $\kappa$-BM problem}

Next consider the $\kappa$-BM problem. In \cref{sec:prelim}, we describe an MCNF formulation when the level intersection sizes $s'_{i_1,i_2,\ldots ,i_P}$ for $p=1,...,P$ and $i_p=1,...,k_p$ are given. Observe that if we treat the sizes $s'_{i_1,i_2,\ldots ,i_P}$ for all $p$ and $i_p$ as decision variables, then by enforcing the integrality of these $K$ variables and adding the constraints saying that

\[ \sum_{i_1=1}^{k_1}...\sum_{i_{p-1}=1}^{k_{p-1}}\sum_{i_{p+1}=1}^{k_{p+1}}...\sum_{i_{P}=1}^{k_{P}}  s'_{i_1,i_2,\ldots ,i_P}=\kappa\cdot \ell_{p,i_p} , \quad  i_p=1,...,k_p\ \ p=1,\ldots ,P  \]
forcing the $\kappa$-fine balance constraints
to the MCNF formulation, we get a MILP formulation of $\kappa$-BM with $K$ integral variables.
In fact if we restrict ourselves to a common integral values of these $K$ variables, then the other decision variables are integral as we argue next.  By considering the values of these $K$ integral variables as constants, the resulting linear programming formulation is in fact an MCNF LP formulation whose supply/demand vector depends on the values of these $K$ integral variables. Thus, the optimal solution for the MILP is without loss of generality integral, and even if it does not satisfy this integral requirement it can be transformed to another optimal solution that is integral in polynomial time.

Since the number of variables of the resulting mixed-integer program is at most $n\cdot n'+K$, the number of integer variables is $K$, and the number of constraints is $O(n\cdot n')$, we conclude that the algorithm that formulates this MILP and solves it to optimality guaranteeing that the optimal solution is integral, runs in $FPT(K)$ time.

\subsection{The MSBM and $\kappa$-MSBM problems}
We know from \cref{thm:NPkMSBM} that the $1$-covariate $\kappa$-MSBM problem for $\kappa\geq 3$ is NP-hard already if the unique covariate has only one level.

We consider next the MSBM problem. In \cref{sec:prelim}, we describe an MCNF formulation if all the level intersection sizes $s_{i_1,i_2,\ldots ,i_P}$ and $s'_{i_1,i_2,\ldots ,i_P}$ are given. Observe that if we treat $s_{i_1,i_2,\ldots ,i_P}$ and $s'_{i_1,i_2,\ldots ,i_P}$ as decision variables, then by enforcing the integrality of these $O(K)$ variables and adding the constraints saying that

\[  \sum_{i_1=1}^{k_1}...\sum_{i_{p-1}=1}^{k_{p-1}}\sum_{i_{p+1}=1}^{k_{p+1}}...\sum_{i_{P}=1}^{k_{P}} s_{i_1,i_2,\ldots ,i_P} = \sum_{i_1=1}^{k_1}...\sum_{i_{p-1}=1}^{k_{p-1}}\sum_{i_{p+1}=1}^{k_{p+1}}...\sum_{i_{P}=1}^{k_{P}}  s'_{i_1,i_2,\ldots ,i_P}, \quad  i_p=1,...,k_p\ \ p=1,\ldots ,P  \]
that is the fine balance constraints,
in addition to the constraint saying that the sum over all $s_{i_1,i_2,\ldots ,i_P}$ equals the objective function value of $\kappa$-FBS,
to the MCNF formulation, we get a MILP formulation of $\kappa$-MSBM with $2K$ integral variables.
In fact if we restrict ourselves to a common integral values of these $2K$ variables, then the other decision variables are without loss of generality integral as well as we argue next.  By considering the values of these $2K$ integral variables as constants, the resulting linear programming formulation is in fact a MCNF LP formulation whose supply/demand vector depends on the values of these $2K$ integral variables. Thus, the optimal solution for the MILP is without loss of generality integral, and even if it does not satisfy this integral requirement it can be transformed to another optimal solution that is integral in polynomial time.

Since the number of variables of the resulting mixed-integer program is at most $n\cdot n'+2K$, the number of integer variables is $2K$, and the number of constraints is $O(n\cdot n')$, we conclude that the algorithm that formulates this MILP and solves it to optimality guaranteeing that the optimal solution is integral, runs in $FPT(K)$ time.

\appendix

	\section{The minimum cost network flow}\label{sec:flow}
	We formulate here the {\em minimum cost network flow problem} (MCNF). The input to the problem is a graph $G=(V,A)$ with a set of nodes $V$ and a set of arcs $A$, where each arc $(i,j)\in A$ is associated with a cost $c_{ij}$, capacity upper bound $u_{ij}$, and capacity lower bound $l_{ij}$. Each node $i\in V$ has supply $b_i$ which is interpreted as demand if negative, and can be $0$.   Let $x_{ij}$ be the amount of flow on arc $(i,j)\in A$.  The flow vector $\x$ is said to be feasible if it satisfies:\\
	(1) {\bf Flow balance constraints:}  For every node $k \in V$  $Outflow(k)-Inflow (k) = b_k$\\ 
	(2) {\bf Capacity constraints:}  For each arc $(i,j)\in A$, $l_{ij} \leq x_{ij} \leq u_{ij}$.
	
	The linear programming formulation of the problem is:
	\[
	\hspace{.4in}\begin{array}{ll}
	\mbox{(MCNF)~~~~}  \min & \sum_{(i,j)\in A}\ c_{ij} x_{ij}  \\
	\mbox{subject to }\ & \sum_{j: (k,j)\in A} x_{kj}\,-\,
	\sum_{i: (i,k)\in A} x_{ik} = b_k  \ \forall k \in V \\
	&     l_{ij} \le x_{ij} \le u_{ij}, \  \ \forall (i,j) \in A.
	\end{array}
	\]

	The flow balance constraints coefficients form a $\{ 0,1,-1\}$-matrix where in each column there is exactly one $1$ and one $-1$.  Such matrix is a special case of matrices where each column (or row) has at most one $1$ and at most one $-1$, which are known to be totally unimodular.

\bibliographystyle{siamplain}

\end{document}